\documentclass[conference]{IEEEtran}

\usepackage{braket}
\usepackage{marvosym}
\usepackage{amsfonts}
\usepackage{amssymb}
\usepackage{amsmath}
\usepackage{amsthm}
\usepackage{dcolumn}
\usepackage{bm,bbm}

\newcommand{\cU}{{\mathcal U}}
\newcommand{\cC}{{\mathcal C}}
\newcommand{\cP}{{\mathcal P}}
\newcommand{\cD}{{\mathcal D}}
\newcommand{\cX}{{\mathcal X}}

\newcommand{\cM}{{\mathcal M}}
\newcommand{\cS}{{\mathcal S}}

\newcommand{\cE}{{\mathcal E}}

\newcommand{\hr}{{\mathcal H}}
\newcommand{\cK}{{\mathcal K}}

\newcommand{\cV}{{\mathcal V}}
\newcommand{\cL}{{\mathcal L}}

\newcommand{\cT}{{\mathcal T}}

\newcommand{\op}{\cC^{\downarrow}}

\newcommand{\pr}{\mathfrak{P}}

\newcommand{\fS}{{\mathfrak S}}

\newcommand{\bbR}{{\mathbb R}}
\newcommand{\bbN}{{\mathbb N}}

\newcommand{\bfS}{{\mathbf S}}

\newcommand{\eins}{\mathbbm 1}

\newcommand{\id}{\mathrm{id}}

\newcommand{\tr}{\mathrm{tr}}

\newcommand{\supp}{\mathrm{supp}}

\newcommand{\sr}{\mathrm{sr}}
\newcommand{\conv}{\mathrm{conv}}

\newtheorem{theorem}{Theorem}

\newtheorem{definition}[theorem]{Definition}
\newtheorem{example}[theorem]{Example}

\newtheorem{lemma}[theorem]{Lemma}

\title{Resource Cost Results for Entanglement Distillation and State Merging under Source Uncertainties}

\author{\IEEEauthorblockN{Holger Boche and Gisbert Jan\ss en}\vspace{0.3cm}
\IEEEauthorblockA{Lehrstuhl f\"{u}r Theoretische Informationstechnik\\Technische Universit\"{a}t
M\"unchen, Germany\\Email:\{boche, gisbert.janssen\}@tum.de}}

\begin{document}
\maketitle

\begin{abstract}
We introduce one-way LOCC protocols for quantum state merging for compound sources, which have asymptotically 
optimal entanglement as well as classical communication resource costs. 
For the arbitrarily varying quantum source (AVQS) model, we determine the one-way entanglement 
distillation capacity, where we utilize the robustification and elimination techniques, well-known from
classical as well as quantum channel coding under assumption of arbitrarily varying noise.
Investigating quantum state merging for AVQS, we demonstrate by example, that the usual robustification
procedure leads to suboptimal resource costs in this case.
\end{abstract}

\section{Introduction}\label{sect:introduction}
Communication tasks on two-party quantum sources have been
investigated with extensive results. Especially protocols restricted to local operations and classical 
communication (LOCC) and potential use of pure entanglement as communication resource are of 
special interest for quantum communication as well as entanglement theory.\newline
Quantum state merging and entanglement distillation, two prominent instances within this paradigm, are
considered in this work. For the asymptotic scenario, where large blocklengths are considered, optimal 
resource cost results for i.i.d. quantum sources with perfectly known bipartite density matrix $\rho_{AB}$ 
have been determined in \cite{devetak05c} for entanglement distillation and in \cite{horodecki07} for 
quantum state merging. Generalizations of these results to the
compound source model, where the source describing density matrix is not perfectly known, but only identified
as a member of a set $\cX$ of states, were partly given in \cite{bjelakovic13}.
While the optimal asymptotic entanglement cost of one-way state merging for compound sources was determined 
in 
\cite{bjelakovic13}, the classical side communication cost of the protocols introduced there was suboptimal in 
general. The present work contributes protocols which are optimal regarding the entanglement as well as 
classical cost. \newline
We mention here, that it seems tractable, to combine techniques from \cite{bjelakovic13} with one-shot 
results for quantum state merging given in \cite{berta09} to establish optimal universal protocols for quantum 
state merging of compound sources also in the regime of finite blocklengths.\newline
From the communication perspective, it is highly desirable, to consider these protocols under more general 
source scenarios.
In this work, we consider the AVQS source model, where the source density matrix is allowed to vary from 
output to output in an arbitrary manner over a set $\cX$ of possible states. This source
model might imagined as result of a powerful communication attack, where an adversarial party is allowed to 
choose any
state from $\cX$ for each output of the source, forcing the communication parties to perform protocols which 
simultaneously work well for each possible output sequence. \newline 
Communication settings with arbitrarily varying channels (AVCs) and sources were first investigated in classical 
Shannon theory, where the famous robustification and elimination techniques introduced by Ahlswede 
\cite{ahlswede78,ahlswede80} were demonstrated to be useful. Considering message transmission under 
the average error criterion, the mentioned techniques allow to derive asymptotically errorless coding schemes 
for a given AVC from coding schemes with exponentially decreasing error for a certain compound channel. 
Concerning channel coding scenarios 
assuming arbitrarily varying quantum channels, 
coding theorems were shown in e.g. in \cite{ahlswede07, ahlswede13}. \newline
In this work, we utilize the robustification and elimination techniques to determine the optimal entanglement
rates for one-way entanglement distillation, and therefore generalize results from \cite{devetak05c} and 
\cite{bjelakovic13} to the AVQS scenario. 
We also consider quantum state merging for AVQS, and demonstrate, that the robustification 
approach to the arbitrarily 
varying setting is of limited usage in this case. We give an example, which shows, that actually, better 
(i.e. lower) entanglement, as well as classical communication rates are possible, than delivered by the 
robustification-based
approach. \newline
Due to space limitations in this paper, we restrict ourselves to brief proof sketches of the results. The full
arguments and further explanations can be gathered in \cite{boche14} accompanying this work. 
\section{Notation and conventions}
All Hilbert spaces appearing in this work are considered to be finite dimensional complex vector spaces. In our
notation $\mathcal{L}(\hr)$ is the set of linear maps and $\cS(\hr)$ the set of states (density matrices) on a 
Hilbert space $\hr$, while we denote the set of quantum channels (i.e. completely positive (cp) and 
 trace preserving maps) from $\mathcal{L}(\hr)$ to $\mathcal{L}(\cK)$ by $\mathcal{C}(\hr, \cK)$ and the set 
 of trace-nonincreasing cp maps by $\mathcal{C}^{\downarrow}(\hr,\cK)$ for two Hilbert spaces $\hr$, $\cK$. 
 Regarding states on multiparty systems, we freely make use of the following convention for a system consisting
 of some parties $X,Y,Z$, for instance, we denote $\hr_{XYZ} := \hr_{X} \otimes \hr_Y \otimes \hr_Z$, and denote
 the marginals by the letters assigned to subsystems, i.e. $\sigma_{XZ} := \tr_{\hr_Y}(\sigma)$ for $\sigma \in
 \cS(\hr_{XYZ})$ and so on. For a bipartite pure state $\ket{\psi}\bra{\psi}$ on a Hilbert space $\hr_{XY}$,
 we denote its Schmidt rank (i.e. number of nonzero coefficients in the Schmidt representation of $\psi$) by 
 $\sr(\psi)$. We use the definition $F(a,b) := \|\sqrt{a}\sqrt{b}\|_1^2$ for matrices $a,b \geq 0$ (F is the 
 quantum fidelity in case that $a,b$ are density matrices). \newline
 The von Neumann entropy is denoted $S(\cdot)$. The usual notation for entropic 
 quantities extended here to indicate state dependency, we write $I(X;Y,\rho)$ ($I_c(X\rangle Y, \rho)$, 
 $S(X|Y,\rho$))
 for the quantum mutual information (coherent information, conditional entropy) of a bipartite state $\rho$
 shared by parties $X$ and $Y$.
 The protocols we consider are build from one-way LOCC channels, which we define concisely in the following 
 (see also the appendix of \cite{bjelakovic13} and references therein). \newline
 A quantum instrument $\mathcal{T}$ on a Hilbert space $\hr$ is given  by a set 
 $\{\mathcal{T}_k\}_{k=1}^K \subset \mathcal{C}^{\downarrow}(\hr,\cK)$ of trace non-increasing cp maps, 
 such that $\sum_{k=1}^K \mathcal{T}_k$ is  a channel. With bipartite Hilbert spaces $\hr_{AB}$ and 
 $\cK_{AB}$, a channel $\mathcal{N} \in \mathcal{C}(\hr_{AB},\cK_{AB})$ is an $A \rightarrow B$ (one-way) 
 LOCC channel, if it is a combination of an instrument $\{\mathcal{T}_k\}_{k=1}^K \subset 
 \mathcal{C}^{\downarrow} (\hr_A,\cK_A)$ and a family $\{\mathcal{R}_k\}_{k=1}^K \subset 
 \mathcal{C}(\hr_B,\cK_B)$ of channels in the sense,  that it can be written in the form
 \begin{align}
  \mathcal{N}(a) = \sum_{k=1}^K (\mathcal{T}_k \otimes \mathcal{R}_k)(a) && (a \in \mathcal{L}(\hr_{AB})).
  \label{locc_definition}
 \end{align}
 The number of different messages which $A$ has to send to $B$ within the application of $\mathcal{N}$ is 
 $K$ (interchanging parties gives the definition of $B \rightarrow A$ LOCC channels). \newline
 We denote the set of classical probability distributions on a set $\bfS$ by $\pr(\bfS)$. The $l$-fold cartesian
 product of $\bfS$ will be denoted $\bfS^l$ and $s^l := (s_1,...,s_l)$ is the notation for elements of 
 $\bfS^l$. For a natural number $n$, the shortcut $[n]$ is used to abbreviate the set $\{1,...,n\}$.
 For a set $A$ we denote the convex hull of $A$ by $\conv(A)$ and its boundary by $\partial A$. By $\fS_l$, we denote the group of permutations 
 on $l$ elements, 
 in this way $\sigma(s^l) = (s_{\sigma(1)},...,s_{\sigma(l)})$ for each $s^l = (s_1,...,s_l)\in \bfS^l$ and 
 permutation $\sigma \in \fS_l$. For any two nonempty sets $\mathcal{X}$, $\mathcal{X}'$ of states on a Hilbert space $\hr$, the Hausdorff
 distance between $\mathcal{X}$ and $\mathcal{X}'$ (induced by the trace norm $\|\cdot\|_1$) is defined by
 \begin{align*}
  &d_H(\mathcal{X},\mathcal{X}') \\
  &:= \max \left\{ \sup_{\sigma \in \mathcal{X}}\inf_{\sigma' \in \mathcal{X}'}\|\sigma - \sigma'\|_1, 
  \sup_{\sigma' \in \mathcal{X}'}\inf_{\sigma \in \mathcal{X}}\|\sigma - \sigma'\|_1 \right\}.
 \end{align*}
 
\section{Basic definitions} \label{sect:protocol}
Let $\cX = \{\rho_s\}_{s \in \bfS}$ be a set of states on a Hilbert space $\hr$. The \emph{(memoryless) 
compound source generated by} $\cX$ is given by the family $\{ \{\rho_s^{\otimes l}\}_{s \in \bfS} \}_{l \in \bbN}$ of states. 
The \emph{arbitrarily varying source (AVQS) generated by} $\mathcal{X}$ is given by the family 
$\{\{\rho_{s^l} \}_{s^l \in \bfS^l} \}_{l \in \bbN}$, with shortcut definitions 
\begin{align*}
  \rho_{s^l} := \rho_{s_1} \otimes  ... \otimes \rho_{s_l} && (s^l = (s_1,...,s_l) \in \bfS^l).
\end{align*}
We identify compound and AVQ sources with their generating sets and 
write \emph{the compound source} $\mathcal{X}$ and \emph{the AVQS} $\mathcal{X}$.
\newline
A quantum channel $\cM$ is called an \emph{$(l,k_l,D_l)$-$A \rightarrow B$ merging for states on $\hr_{AB}$}
\cite{horodecki07} if it is an $A \rightarrow B$ one-way LOCC 
\begin{align*}
 \mathcal{M}: \cL(\cK_{0,AB}^l \otimes \hr_{AB}^{\otimes l}) \rightarrow 
 \cL(\cK_{1,AB}^l \otimes \hr_{B'B}^{\otimes l}),
\end{align*}
with $D_l$ summands as in defined in (\ref{locc_definition}),
where $\hr_{B'} \simeq \hr_{A}$ is a Hilbert space under control of $B$, and
$\cK_{AB,0}^l, \cK_{AB,1}^l$ are bipartite Hilbert spaces of systems shared by $A$  and $B$. 
These spaces are reserved to carry input and output maximally entangled states $\phi_0^l$ and $\phi_1^l$, 
which we assume to have maximal Schmidt rank, such that  
\begin{align*}
 k_l := \frac{\dim \cK_{A,0}^l}{\dim \cK_{A,1}^l} = 
 \frac{\dim \cK_{B,0}^l}{\dim \cK_{B,1}^l} = \frac{\sr(\phi_0^l)}{\sr(\phi_1^l)}. 
\end{align*}
holds. Given a state $\rho$ on $\hr_{AB}^{\otimes l}$, and an $(l,k_l,D_l)$ $A \rightarrow B$ merging 
$\mathcal{M}_l$,
the measure of fidelity of the protocol applied to $\rho$ is defined
\begin{align}
 F_m(\rho,\mathcal{M}_l) := F\left(\mathcal{M}_l \otimes \id_{\hr_{E}^l}(\phi_0^l \otimes \psi), \phi_1^l 
 \otimes \psi'\right). \label{merging_fid_def}
\end{align}
In (\ref{merging_fid_def}), $\psi$ is any purification of $\rho$ on $\hr_{AB}^{\otimes l} \otimes \hr_{E}^l$
with an additional Hilbert space $\hr_{E}^l$, and $\psi'$ is a version of the state $\psi$ on 
$\hr_{B'B}^{\otimes l} \otimes \hr_E^l$. It was shown in \cite{bjelakovic13}, that the r.h.s. of the equality
(\ref{merging_fid_def}) is independent of the choice of purification.
\begin{definition}
 A real number $R_q$ is called an \emph{achievable entanglement cost} for $A \rightarrow B$ merging 
 of the AVQS  $\cX \subset \cS(\hr_{AB})$ with classical communication rate $R_c$, 
 if there exists a sequence 
 $\{\cM_l\}_{l \in \bbN}$ of $(l,k_l,D_l)$-$A \rightarrow B$-mergings, which fulfills the conditions
 \begin{enumerate}
  \item $\underset{l \rightarrow \infty}{\lim}\ \underset{s^l \in \bfS^l}{\inf}\ F_m(\rho_{s^l},\cM_l) = 1$
  \item $\underset{l \rightarrow \infty}{\limsup}\ \frac{1}{l} \log k_l \leq R_q$
  \item $\underset{l \rightarrow \infty}{\limsup}\ \frac{1}{l} \log D_l \leq R_c.$
   \end{enumerate}
\end{definition}
The corresponding definition regarding achievable entanglement costs for the compound source $\cX$ can 
be easily guessed, where the first condition in the above definition has to be replaced by
\begin{enumerate}
 \item[1')] $\underset{l \rightarrow \infty}{\lim}\ \underset{s\in \bfS}{\inf}\ F_m(\rho_{s}^{\otimes l}
 ,\cM_l) = 1.$
\end{enumerate}
\begin{definition}
 The \emph{$A \rightarrow B$-one-way merging cost} $C_{m,\rightarrow}^{AV}(\cX)$ of the AVQS $\cX$ is 
 defined by
 \begin{align*}
  C_{m,\rightarrow}^{AV}(\cX) := \inf\left\{R_q:  \begin{array}{l}
	      R_q \ \textrm{is achievable entanglement }\\ 
	      \textrm{cost for}\ A \rightarrow B \ \textrm{merging of the} \ \\
	        \textrm{AVQS}\ \cX \ \textrm{with some classical} \\
	        \textrm{communication rate}\ R_c\in \bbR 
	  \end{array} \right\}
 \end{align*}
 
\end{definition}
The \emph{$A \rightarrow B$ merging cost for merging of the compound source $\cX$} is 
 defined analogously and denoted $C_{m,\rightarrow}(\cX)$ \cite{bjelakovic13,boche14}.
 Concerning entanglement distillation, we are interested in the entanglement gain of one-way LOCC distillation 
procedures. 
\begin{definition}
 A non-negative number R is an \emph{achievable $A \rightarrow B$ entanglement distillation rate for the AVQS
 $\cX$ with classical communication rate $R_c$}, 
 if there exists a sequence $\{\cD_l\}_{l \in \bbN}$ of $A \rightarrow B$ LOCC channels, each with a  
 representation as given in (\ref{locc_definition}) with $D_l$ summands,  such 
 that the conditions 
 \begin{enumerate}
  \item $\underset{l \rightarrow \infty}{\lim} \ \inf_{s^l \in \bfS^l} F(\cD_l(\rho_{s^l}), \phi_l) = 1$
  \item $\underset{l \rightarrow \infty}{\liminf} \frac{1}{l} \log \sr(\phi_l) \geq R$
  \item $\underset{l \rightarrow \infty}{\limsup} \frac{1}{l} \log D_l \leq R_c$
 \end{enumerate}
 hold, where $\phi_l$ is a maximally entangled state shared by $A$ and $B$ for each $l \in \bbN$.
\end{definition}
 \begin{definition}
  The \emph{$A \rightarrow B$ entanglement distillation capacity for the AVQS} $\cX$ is defined
   \begin{align*}
    D^{AV}_{\rightarrow}(\mathcal{X}) := 
    \sup\left\{R : \begin{array}{l} R \ \textrm{is achievable}\ A \rightarrow B\ dis- \\ 
			\textrm{tillation rate for the AVQS } \mathcal{X} \\
			  \textrm{with some classical rate}\ R_c\end{array} \right\}.
   \end{align*}
 \end{definition}
 For entanglement distillation again, the definitions in case of a compound quantum source can be easily 
 guessed, and we denote the \emph{one-way entanglement distillation
 capacity for the compound source} $\cX$ by $D_{\rightarrow}(\cX)$. We do not determine optimal classical
 communication rates for entanglement distillation here. These are unknown in general even in the case 
 where the source is i.i.d. with perfectly known state \cite{devetak05c}.
 \section{Quantum State Merging for Compound Sources}
 In this section, we show existence of $A \rightarrow B$ LOCC protocols, which are asymptotically optimal 
 regarding their entanglement as well as classical side communication requirements, due to the converse
 results given in \cite{bjelakovic13}. Optimality is known from the corresponding converse statement given in
 \cite{bjelakovic13}, Section V, where it was shown, that successful one-way merging of a compound source
 $\mathcal{X}$ is not possible with asymptotic classical cost smaller than $\sup_{\rho \in \mathcal{X}}
 I(A;E, \rho)$.
 \begin{theorem} \label{main_th_comp_merg}
 Let $\cX \subset \cS(\hr_{AB})$ be a set of bipartite states. For each $\delta > 0$, 
 $R_q = \sup_{\rho \in \cX} S(A|B,\rho) + \delta$ is an achievable entanglement cost for $A \rightarrow B$ 
 merging of the compound source $\mathcal{X}$ with classical communication rate
 \begin{align}
  R_c = \sup_{\rho \in \cX} I(A;E,\rho) + \delta, \label{comp_opt_cl}
 \end{align}
 where $I(A;E,\rho) = S(\rho_A) + S(A|B,\rho)$ is the quantum mutual information between the $A$ and $E$ 
 systems of a purification of $\rho$ on a larger system with parties $A,B,E$.
\end{theorem}
 Our proof of Theorem \ref{main_th_comp_merg} has two main ingredients. We use slight generalizations of 
 the results from \cite{bjelakovic13}, Theorems 4 and 6 (see \cite{boche14}) where achievability of the 
 entanglement cost $\sup_{\rho \in \cX} S(A|B,\rho)$ was shown. However, the protocols used there, have 
 classical $A\rightarrow B$ communication requirements of at least 
  $\sup_{\rho \in \cX} S(\rho_A) + \sup_{\rho \in \cX} S(A|B,\rho)$
 which is, in general, strictly greater than the rate given in (\ref{comp_opt_cl}).
 We show, that $R_q$ is achievable with classical communication of rate $R_c$ from (\ref{comp_opt_cl}) 
 by combining the protocols from \cite{bjelakovic13} with a suitably fine-grained entropy estimating 
 instrument on the $A$-system.
 \begin{proof}[Proof of Theorem \ref{main_th_comp_merg}]
 Let $\delta > 0$, $l \in \bbN$, $d:= \dim \hr_A$ and assume $\tilde{s} := \sup_{\rho \in \cX} S(A|B,\rho) 
 - \frac{\delta}{2} < 0$ (the remaining case $\tilde{s} \geq 0$ follows by simple modifications). Consider the sequence
 $s_0 := 0 < s_1 < ... < s_N := \log d$ with $s_i := s_{i-1} + \eta$, $ 1 \leq i < N$, and the intervals
 $I_0 := [s_0,s_1], I_i := (s_{i-1},s_i)$, $1 < i < N$. These generate a decomposition
 $\cX_1,...,\cX_N$ of $\cX$ into pairwise disjoint subsets (some may be empty), defined
 \begin{align*}
  \cX_i := \{\rho \in \cX: \ S(\rho_A) \in I_i\} && (i \in [N]).
 \end{align*}
 We further define sets $\tilde{\cX}_i := \bigcup_{j \in n(i)} \cX_j$,
 where $n(i) := \{j \in [N]: |j-i|\leq 1\}$. We construct an entropy estimating instrument
 $\{\cP^{(i)}\}_{i \in [N]} \subset \op(\hr_A^{\otimes l}, \hr_{A}^{\otimes l})$ by defining
 \begin{align*}
  \cP^{(i)} := p_i (\cdot) p_i, \ \text{with} \hspace{0.2cm} p_i 
  := \sum_{\substack{\lambda: H(\overline{\lambda}) \in I_i}} P_{\lambda,l} &&(i \in [N]),
 \end{align*}
  where $P_{\lambda,l}$ is the projection supported on the invariant subspace of $\hr_{A}^{\otimes l}$ 
  belonging to the representation of $\fS_l$ labeled by Young frame $\lambda$, and $H(\overline{\lambda})$ 
  is the Shannon entropy of the probability distribution given by the normalized box-lengths $\lambda$ 
  \cite{christandl06}. 
  It can be shown (using the bounds from Theorem 1 in \cite{christandl06}, which first appeared in \cite{keyl01}), 
  that our definitions imply for sufficiently large blocklength $l$,
 \begin{align}
  \sum_{j \in [N]\setminus n(i)} \tr(\cP^{(j)}\otimes id_{\hr_B^{\otimes l}}(\rho^{\otimes l}))
   \leq 2^{-lc_2} \label{entropy_est}
 \end{align}
  for each $i \in [N]$ and $\rho \in \cX_i$ with a positive constant $c_2 = c_2(\eta)$. Moreover it is
  known from \cite{bjelakovic13}, Theorem 6 (with some simple modifications, see \cite{boche14}), 
  that for each
  $i$ with $\tilde{\cX}_i \neq \emptyset$ and large enough blocklength, there is a $(l,k_l,
  \tilde{D}^{(i)}_l)-A\rightarrow B$ merging $\tilde{\cM}^{(i)}$ with 
  \begin{align}
   \inf_{\rho \in \tilde{\cX}_i} F_m(\rho^{\otimes l},\tilde{\cM}_l^{(i)}) \geq 1 - 2^{-lc_3} 
   \label{merg_sub_rt}
  \end{align}
  with a positive constant $c_3 > 0$, $k_l \geq 2^{-l\tilde{s}}$ and, for each $i$,
  \begin{align}
   \frac{1}{l} \log \tilde{D}_l^{(i)} 
   &\leq \sup_{\rho \in \tilde{\cX}_i} S(\rho_A) 
   + \sup_{\rho \in \tilde{\cX}_i} S(A|B,\rho) + \frac{\delta}{2} \\
   &\leq \sup_{\rho \in \tilde{\cX}_i} I(A;E,\rho) + \frac{\delta}{2} + 3 \eta. \label{class_bums}
  \end{align}
  Define
  \begin{align*}
   \cM_l := \sum_{i \in [N]} \tilde{\cM}_l^{(i)} \circ (\cP^{(i)} \otimes id_{\hr_B^{\otimes l}}), 
  \end{align*}
  and observe, that $\cM_l$ is an $(l,k_l,D_l)$ $A \rightarrow B$ merging with 
  \begin{align*}
   \frac{1}{l}\log D_l = \frac{1}{l}\log\left(\sum_{i=1}^N D_i\right) 
   \leq \sup_{\rho \in \cX} I(A;E,\rho) + \frac{\delta}{2} + 3 \eta.
  \end{align*}
  Eqns (\ref{entropy_est}), (\ref{merg_sub_rt}) and properties of the merging fidelity imply
  \begin{align}
   \inf_{\rho \in \cX} F_m(\rho^{\otimes l}, \cM_l) \geq 1 - 2^{-l c_4}
  \end{align}
  with a positive constant $c_4$ for large enough blocklength. Since $\eta$ and $\delta$ are free 
  to choose, we are done.
 \end{proof}

\section{One-way Entanglement Distillation for AVQS}
The following theorem determines the capacity for $A \rightarrow B$ one way entanglement distillation in 
presence of an AVQS generated by a set $\cX$ of density matrices. As in several coding theorems for classical
AV channels and sources, the capacity of the AV source $\cX$ equals the capacity of the compound source 
$\overline{\conv(\cX)}$. Notice, that it makes no difference to consider $\conv(\cX)$ instead of its 
closure since these sets have Hausdorff distance zero and the capacity function is continuous 
(see \cite{boche14}).
\begin{theorem} \label{main_th_av_dist}
 Let $\cX \subset \cS(\hr_{AB})$ be a set of bipartite states. It holds
 \begin{align*}
  D_{\rightarrow}^{AV}(\cX) 
  &= D_{\rightarrow}(\conv(\cX)) \\
  &= \lim_{k \rightarrow \infty} \frac{1}{k}\max_{\cT} 
  \inf_{\tau \in \conv(\cX)} D^{(1)}(\tau^{\otimes k},\cT), 
 \end{align*}
 where the maximization is over all instruments with domain $\cL(\hr_X)$, and for each state $\sigma$ on a 
 bipartite space $\hr_{XY}$ and instrument $\cE = \{\cE_j\}_{j=1}^J$, we use 
 the definition 
 \begin{align*}
  D^{(1)}(\sigma,\cE) 
  := \sum_{\substack{\lambda_j(\sigma):\\ \lambda_j \neq 0}} \lambda_j(\sigma) \ I_c(X\rangle Y, \hat{\sigma}_j)
 \end{align*}
 with $\lambda_j(\sigma) := \tr(\cE_j \otimes \id_{\hr_Y}(\sigma))$ and 
 $\hat{\sigma}_j := \lambda_j(\sigma)^{-1} \cE_j \otimes \id_{\hr_Y}(\sigma)$. 
\end{theorem}
 In the proof of Theorem \ref{main_th_av_dist} below, we invoke the robustification technique \cite{ahlswede80}, 
 to generate good entanglement distillation schemes for the AVQS $\cX$ from good protocols for the compound 
 source $\overline{\conv{\cX}}$.
 \begin{lemma}[Robustification technique, cf. \cite{ahlswede80} and Theorem 6 in \cite{ahlswede86}]\label{robustification-technique}
Let $\bfS$ be a set with $|\bfS|<\infty$ and $l\in\bbN$. If a function $f:\bfS^l\to [0,1]$ satisfies
\begin{align}\label{eq:robustification-1}
 \sum_{s^l\in\bfS^l}f(s^l)q(s_1)\cdot\ldots\cdot q(s_l)\geq 1-\gamma
\end{align}
for each type $q$ of sequences in $\bfS^l$ for some $\gamma\in [0,1]$, then
\begin{align}\label{eq:robustification-2}
  \frac{1}{l!}\sum_{\sigma\in \mathfrak{S}_l}f(\sigma(s^l))\ge 1-(l+1)^{|\bfS  |}\cdot \gamma\qquad \forall s^l\in \bfS^l.
\end{align}
\end{lemma}
 
\begin{proof}[Proof of Theorem \ref{main_th_av_dist}]
 To show achievability, we first prove the assertion of the theorem for the case of a finite set
 $\mathcal{X}:=\{\rho_s\}_{s\in \bfS}$. We show, that each achievable 
 $A \rightarrow B$ 
 entanglement distillation rate for the compound source $\conv(\cX)$ is also achievable for the AVQS $\cX$
 and use the fact, that
 \begin{align*}
  \conv(\cX) = \left\{\rho_p: \rho_p = \sum_{s \in \bfS} p(s) \rho_s, \ p \in \pr(\cX)\right\}
 \end{align*}
 holds. Assuming, that $R$ is an achievable rate for the compound source $\conv(\cX)$, we know,
 that for each $\delta > 0$ and large enough blocklength $l$, there is an $A \rightarrow B$ LOCC channel
 $\cD_l$, such that for each $p \in \pr(\bfS)$ the fidelity is bounded 
 $F(\cD_l(\rho_p^{\otimes l}), \phi_l) \geq 1 - 2^{-lc_5}$ with a constant $c_5 > 0$ (in the proof of 
 Theorem 8 in \cite{bjelakovic13}, it was shown, that each achievable rate can be achieved by protocols
 with exponentially decreasing error).
 With $f(s^l) := \rho_{s^l}$, and linearity of the fidelity in the first input, we yield
 \begin{align}
  \sum_{s^l} p^l(s^l)\ f(s^l) \geq 1-2^{-lc_5} \hspace{0.2cm} \text{for all} \hspace{0.2cm} p \in \pr(\bfS). 
  \label{rob_condition}
 \end{align}
 Eq. (\ref{rob_condition}) implies, that the function $f$ satisfies the conditions of 
 Lemma \ref{robustification-technique}. Let 
 $\cU_\sigma$ be the (local) unitary channel, which permutes the tensor factors in $\hr_{AB}^{\otimes l}$ 
 according to the permutation $\sigma$, i.e. $\rho_{\sigma(s^l)} = \cU_\sigma(\rho_{s^l})$, and 
 $f(\sigma(s^l)) = F(\cD_l \circ \cU_\sigma(\rho_{s^l}), \phi_l)$. 
 Lemma \ref{robustification-technique} and (\ref{rob_condition}) imply, that 
 \begin{align*}
  1-(l+1)^{|\bfS|}\cdot 2^{-lc_5} 
   \leq \frac{1}{l!} \sum_{\sigma \in l} f(\sigma(s^l))
   = F(\hat{\cD}_l(\rho_{s^l}), \phi_l) 
 \end{align*}
 holds for each $s^l$, where we defined an $A \rightarrow B$ one-way LOCC channel $\hat{\cD}_l$ by 
 $\hat{\cD}_l := (l!)^{-1} \sum_{\sigma \in \fS_l} \cD_l \circ \cU_{\sigma}$. From 
 the above inequality and the polynomial growth of the function $(l+1)^{|\bfS|}$ for 
 $l \rightarrow \infty$, we infer, that $R$ is an achievable rate for one-way entanglement distillation 
 for the AVQS $\cX$ with fidelity going to one exponentially fast. By Theorem 8 in \cite{bjelakovic13} 
 (generalized to the case of infinite compound sources in \cite{boche14}), we can choose any rate $R \geq 0$ 
 with $R \leq D_{\rightarrow}(\conv(\mathcal{X}))$. \newline
 However, the protocols $\{\hat{\cD}_l\}_{l=1}^\infty$ we introduced, 
 are not reasonable regarding the classical side communication cost, since $A$ has to communicate the messages
 required within application of $\tilde{\cD}_l$ \emph{and} the choice of
 permutation $\sigma$ (out of $l!$ possibilities), i.e. the classical 
 communication requirements are not rate-bounded for $l \rightarrow \infty$. However, we can invoke the 
 derandomization 
 technique from (\cite{ahlswede78}) to derive protocols with rate-bounded classical communication 
 (see \cite{boche14} for details). \newline 
 To prove the general case of a not necessary
 finite or countable set $\cX$, we apply a polytope approximation technique similar to the one used in  
 \cite{ahlswede13}. For simplicity, we assume $\overline{\conv(\cX)} \cap \partial\cS(\hr) = \emptyset$ 
 (this condition can be removed by slight depolarization of the states in the set $\cX$). Then, for 
  any small enough number $\eta > 0$, we
  find a polytope $P_\eta$, i.e. the convex hull of a finite set $\{\tau_e\}_{e \in E}$ of extreme points,
  such that $\overline{\conv(\cX)} \subset P_\eta \subset \cS(\hr_{AB}) \setminus \partial \cS(\hr_{AB})$, 
  and 
  \begin{align}
  d_H(\conv(\cX),P_{\eta}) < \eta. \label{av_dist_prf_hausdorff}
  \end{align}
  Applying the argument for finite sets given above to $P_\eta$, we find, for each sufficiently large 
  $l$, a distillation protocol $\hat{\cD}_l$, such that 
  \begin{align}
   \min_{e^l \in E^l} F(\hat{\cD}(\tau_{e^l}), \phi_l) \geq 1 - 2^{-lc_5} \label{av_dist_prf_gen_fid}
  \end{align}
  holds with a maximally entangled state $\phi_l$, such that
  \begin{align}
   \frac{1}{l} \log \sr(\phi_l) \geq \frac{1}{k} \max_{\cT}\inf_{\tau \in P_\eta} D(\tau^{\otimes k}, \mathcal{T}) 
   - \frac{\delta}{2} \label{av_dist_prf_gen_rate}.
  \end{align}
  Since $\rho_s$ can be written as a convex combination of extremal points of $P_\eta$, 
  (\ref{av_dist_prf_gen_fid}) implies 
  \begin{align}
   \inf_{s^l \in \bfS^l} F(\hat{\cD}(\rho_{s^l}), \phi_l) \geq 1 - 2^{-lc_5} \label{av_dist_prf_gen_fid_2}.
  \end{align}
  By continuity properties of the function $D^{(1)}$ (see \cite{boche14} for details), 
  and (\ref{av_dist_prf_hausdorff}) together with a (sufficiently small) choice of the parameter $\eta$, 
  it holds
  \begin{align*}
   \max_{\cT} \inf_{\tau \in P_\eta} D^{(1)}(\tau^{\otimes k}, \mathcal{T}) 
   \geq \max_{\cT}\inf_{\rho \in \conv(\cX)}  D^{(1)}(\rho^{\otimes k},\cT) - \frac{k\delta}{2}
  \end{align*}
  which, together with (\ref{av_dist_prf_gen_fid_2}) and (\ref{av_dist_prf_gen_rate}) gives achievability.
  The converse is obvious, since each entanglement distillation protocol $\cD_l$ which is $\epsilon$-good for
  the AVQS $\cX$ is also $\epsilon$-good for entanglement generation of $\overline{\conv(\cX)}$, so that
  the converse for the compound distillation theorem (\cite{bjelakovic13}, Theorem 8) applies.
  \end{proof}
\section{Quantum State Merging for AVQS}
Regarding the task of one-way quantum state merging, the close connection between the merging cost of an 
AVQS $\cX$ and the merging cost of the compound source generated by $\conv(\cX)$ breaks down.
We demonstrate this by example.
\begin{example} \label{main_ex_avqs_merg}
 There exists a set $\mathcal{X}$, such that
 \begin{align*}
  C_{m,\rightarrow}^{AV}(\cX) < C_{m,\rightarrow}(\conv(\cX)).
 \end{align*}
 \end{example}
 Consider the set $\hat{\cX} := \{\rho_s\}_{s=1}^N\subset \cS(\hr_{AB})$, $N < \infty$, with
 \begin{align}
  \rho_s := (U_s \otimes \eins_{\hr_B}) \rho_1 (U_s^\ast \otimes \eins_{\hr_B})  &&(s \in [N]) 
  \label{merg_ex_prf_st}
 \end{align}
 with $\rho_1 \in \cS(\hr_{AB})$ such that $S(A|B,\rho_1) < 0$, and 
 unitary matrices $U_1 := \eins_{\hr_A}, U_2,...,U_N$, such that the supports of the $A$ marginals of the 
 states in $\hat{\cX}$ are pairwise orthogonal. We assume $\dim \hr_A \geq N \cdot \supp(\rho_{A,1})$.
 These definitions also imply for each $s,s' \in [N], s\neq s'$
 \begin{align}
  \rho_{B,s} = \rho_{B,1},\ \text{and} \hspace{0.3cm}\supp(\rho_s) \perp \supp(\rho_{s'}). 
  \label{merg_ex_prf_set}
 \end{align}
 In the following, we show, that for each set constructed in the above manner, it holds the relation
 $
  C_{m,\rightarrow}^{AV}(\hat{\cX}) \leq C_{m,\rightarrow}(\hat{\cX}) - \log N
 $
 for the one-way merging cost. Moreover, each achievable entanglement cost can be achieved with classical 
 communication rate $R_c$ such that
 $
  R_c \leq \sup_{\rho \in \conv(\hat{\cX})} I(A;E,\rho) - \log N
 $
 holds. From the orthogonality conditions (\ref{merg_ex_prf_st}) follows, that there is an instrument 
 $\{\tilde{\cV}_s\}_{s=1}^N$ on $A$'s system, such that 
 \begin{align}
  \cV_{s'}(\rho_s) := \tilde{\cV}_{s'}\otimes \id_{\hr_B}(\rho_s) = \delta_{ss'}\rho_1 \label{merg_ex_prf_instr}
 \end{align}
 holds for each $s \in [N]$. Since $C_{m,\rightarrow}(\rho_1) = S(A|B,\rho_1)$ \cite{horodecki07}, we find
 for each $\delta > 0$ and large enough 
 blocklength $l$, an $(l,k_l,\tilde{D}_l)$ merging for $\rho_1$, with
 \begin{align}
  &k_l \leq 2^{l(S(A|B,\rho_1) + \delta)}, \hspace{0.3cm}
  \tilde{D}_l \leq 2^{l(I(A;E,\rho_1) + \delta)} \label{merg_ex_prf_rate}\\
  \text{and} \hspace{0.3cm}
  &F_m(\rho_1^{\otimes l},\tilde{\cM}_l) \geq 1 - 2^{-l\tilde{c}} \label{merg_ex_prf_error}
 \end{align}
 with a constant $\tilde{c} > 0$. Define $\cM_l := \sum_{s =1}^N \cU_{s^l} \circ \tilde{\cM}_l \circ \cV_{s^l}$,
 with $\mathcal{U}_{s^l}(\cdot)
 := U_{s^l}\otimes \eins_{\hr_B^{\otimes l}}(\cdot)U_{s^l}^\ast \otimes \eins_{\hr_B^{\otimes l}}$. 
 It holds
 \begin{align*}
  &F_m(\rho_{s^l}, \cM_l) \\
  &= \sum_{s'^l\in [N]^l} F(\cU_{s'^l} \circ \tilde{\cM}_l \circ \cV_{s'^l} \otimes 
      \id_{\hr_{E}^{\otimes l}}(\psi_{s^l}), \phi_l \otimes \psi_{s^l}')\\
  &=\sum_{s'^l \in [N]^l} F(\tilde{\cM}_l \otimes \id_{\hr_{E}^{\otimes l}}(\cV_{s'^l}(\psi_{s^l})),
     \phi_l \otimes \cU_{s'^l}^\ast(\psi_{s^l}'))   \\
  &= F_m(\rho_1^{\otimes l}, \tilde{\cM}_l) \geq 1 - 2^{l\tilde{c}}.
 \end{align*}
 for each $s^l \in [N]^l$. The first equality above is by linearity of the merging fidelity in the merging operation,
 the second one is by invariance of the fidelity under unitary evolutions, the third equality is by 
 (\ref{merg_ex_prf_instr}). The last inequality is (\ref{merg_ex_prf_error}). 
 $\cM_l$ is an $(l,k_l,D_l)$-$A\rightarrow B$-merging with $D_l = N^l \cdot \tilde{D}_l$, i.e.
 \begin{align}
  \frac{1}{l}\log D_l \leq I(A;E,\rho_1) + \log N. \label{merg_ex_prf_class_2}
 \end{align}
 By properties of the set $\tilde{\cX}$, i.e. (\ref{merg_ex_prf_st}) and (\ref{merg_ex_prf_set}), and 
 the equality 
 $S(\rho_p) = \sum_{s=1}^N p(s) S(\rho_s) + H(p)$
 which holds for each $\rho_p := \sum_{s=1}^N p(s) \rho_s$, $p \in \pr([N])$, due to orthogonality of the 
 supports of the states, we infer by calculation of the entropies and maximization over $p$
 \begin{align}
  \max_{p \in \pr([N])} S(A|B,\rho_p) &= S(A|B,\rho_1) + \log N, \ \textrm{and} \label{merg_ex_prf_rate_1}\\
  \max_{p \in \pr([N])} I(A;E,\rho_p) &= I(A;E,\rho_1) + 2\log N \label{merg_ex_prf_rate_2}.
 \end{align}
 Eqns. (\ref{merg_ex_prf_rate_1}) and (\ref{merg_ex_prf_rate_2}) together with 
 (\ref{merg_ex_prf_rate}) and (\ref{merg_ex_prf_class_2}), show, that 
 \begin{align*}
  R_q = \max_{p \in \pr([N])} S(A|B,\rho_p) - \log N  
 \end{align*}
 is achievable with asymptotic classical side communication at rate 
 \begin{align*}
  R_c = \max_{p \in \pr([N])} I(A;E,\rho_p) - \log N.
 \end{align*}
\section*{Acknowledgments}
The authors are grateful to Igor Bjelakovi\'c for many stimulating discussions and valuable suggestions.
The work of H.B. is supported by the DFG via grant BO 1734/20-1 and by the BMBF via grant 16BQ1050.

\end{document}